\RequirePackage{fix-cm}
\documentclass[smallextended]{svjour3}       
\smartqed  

\usepackage{hyperref}
\usepackage{wrapfig}
\usepackage{graphicx}
\usepackage[english]{babel}
\usepackage{amsmath}
\usepackage{amssymb}
\usepackage{color}
\usepackage{complexity}
\usepackage{subfig}
\usepackage[export]{adjustbox}
\usepackage{placeins}
\usepackage{comment}
\usepackage{algpseudocode}
\usepackage{float}
\newfloat{algorithm}{t}{lop}
\usepackage[symbol]{footmisc}
\usepackage[round,authoryear]{natbib}

\floatname{algorithm}{Listing}

\algnewcommand\algorithmiclet{\textbf{let}}
\algnewcommand\Let[1]{\algorithmiclet\ #1}

\algnewcommand\algorithmiceach{\textbf{each}}
\algnewcommand\Each[1]{\algorithmiceach\ #1}

\algnewcommand\algorithmicupdate{\textbf{update}}
\algnewcommand\Update[1]{\algorithmicupdate\ #1}

\begin{document}

\title{Efficiently Computing the Shapley Value of Connectivity Games in Low-Treewidth Graphs\thanks{This article shares material with parts of the first author's thesis~\citep{zandthesis}.}}

\titlerunning{Computing the Shapley Value of Connectivity Games in Low-Treewidth Graphs}

\author{Tom C. van der Zanden\and Hans L. Bodlaender \and Herbert J.M. Hamers 
}

\authorrunning{van der Zanden et al.} 

\institute{T.C. van der Zanden \at
              Department of Data Analytics and Digitalisation, Maastricht University, Maastricht, The Netherlands \\
              \email{T.vanderZanden@maastrichtuniversity.nl}
           \and
           H.L. Bodlaender \at
              Department of Information and Computing Science, Utrecht University, Utrecht, The Netherlands \\
              \email{H.L.Bodlaender@uu.nl}
           \and
           H.J.M. Hamers \at
              Department EOR (TiSEM) and TIAS, Tilburg University, Tilburg, The Netherlands \\
              \email{H.J.M.Hamers@tilburguniversity.edu}
}

\date{Received: date / Accepted: date}

\maketitle

\begin{abstract}
The Shapley value is the solution concept in cooperative game theory that is most used in both theoretical as practical settings. Unfortunately, computing the Shapley value is computationally intractable in general.
This paper focuses on computing the Shapley value of (weighted) connectivity games. For these connectivity games, that are defined on an underlying (weighted) graph, computing the Shapley value is $\#\P$-hard, and thus (likely) intractable even for graphs with a moderate number of vertices.
We present an algorithm that can efficiently compute the Shapley value if the underlying graph has bounded treewidth.
Next, we apply our algorithm to several real-world (covert) networks. We show that our algorithm can compute exact Shapley values for these networks quickly, whereas in prior work these values could only be approximated using a heuristic method.
Finally, it is shown that our algorithm can also compute the Shapley value time efficiently for several larger (artificial) benchmark graphs from the PACE 2018 challenge.
\end{abstract}

\section*{Delaractions}

\emph{Funding.} No funds, grants, or other support was received. A large part of this research was done while the first author was associated with Universiteit Utrecht.

\emph{Conflicts of interest/Competing interests.} The authors have no conflicts of interest to declare that are relevant to the content of this article.

\emph{Code availability.}  We have made our source code, as well as the graphs/data used for the experiments, available on GitHub \citep{ShapleyTreewidth}.

\emph{Availability of data and material.} We have made our source code, as well as the graphs/data used for the experiments, available on GitHub \citep{ShapleyTreewidth}.

\clearpage

\section{Introduction}

Motivated by the need to identify important vertices in networks (graphs), many measures for ranking vertices have been suggested. Among these are the classical centrality measures, such as betweenness, closeness and degree \citep[see e.g.,][]{wassermanfaust}.
However, the drawback is that all these measures only take the structure of the network into account.
 In light of this, \emph{game-theoretic centrality} measures have received considerable interest.
These measures take not only the structure of the network into account, but enables also to include special information with respect to individuals and links in the network. Further, such a measure takes also into account the coalitional strength of members in the network.

The Shapley value \citep{shapleyvalue} is most used as centrality measure in these situations. This is not surprising since the Shapley value satisfies intuitive properties \citep[cf.][]{shapleyvalue,Young1985} that are also considered as fair in many situations in practice (e.g., airport landing fees \citep{Littlechild1973}, water transportation \citep{Deidda2009}, genetics \citep{Moretti2007} and terrorism \citep{lindelauf13}).
More applications of the Shapley value in different fields can be found in the survey of \citet{Moretti2008}.

The biggest challenge in applying the Shapley value to real-world applications is the time required to compute it, which generally increases exponentially with the number of players. More precisely, for many games it can be shown that computing the Shapley value is a $\#\P$-hard problem \citep{Faigle1992,michalak2013efficient,michalak2013computational}.


In this paper we focus on the Shapley value for network games in which the vertices are associated with players. In particular we consider connectivity games, introduced by \citet{amer}. These are
 $\{0,1\}$-valued games in which a coalition has the value $1$ if the subgraph induced by this coalition is connected, and has value $0$ otherwise. In a vertex-weighted connectivity game, introduced by \citet{lindelauf13}, the value equals the sum of the vertex weights of a coalition that induces a connected subgraph, and $0$ if the subgraph is not connected.
 
   Lindelauf et al. studied these centrality measures in the context of identifying the most important vertices in terrorist networks \citep{lindelauf13,approximation}. They considered two networks: one (due to \citet{koschade}) consisting of 17 terrorists involved in a 2002 bombing in Bali, the second \citep[due to][]{krebs} consisting of 69 terrorists involved in planning and executing the 9/11 attacks. Whereas for the first network they were able to compute the exact Shapley values, for the second networks this was infeasible and they considered only the part of the network made up by the 19 hijackers that actually carried out the attack. This network in shown in Figure~\ref{fig:911network}.

\begin{figure}
    \centering
    \includegraphics[width=0.95\textwidth]{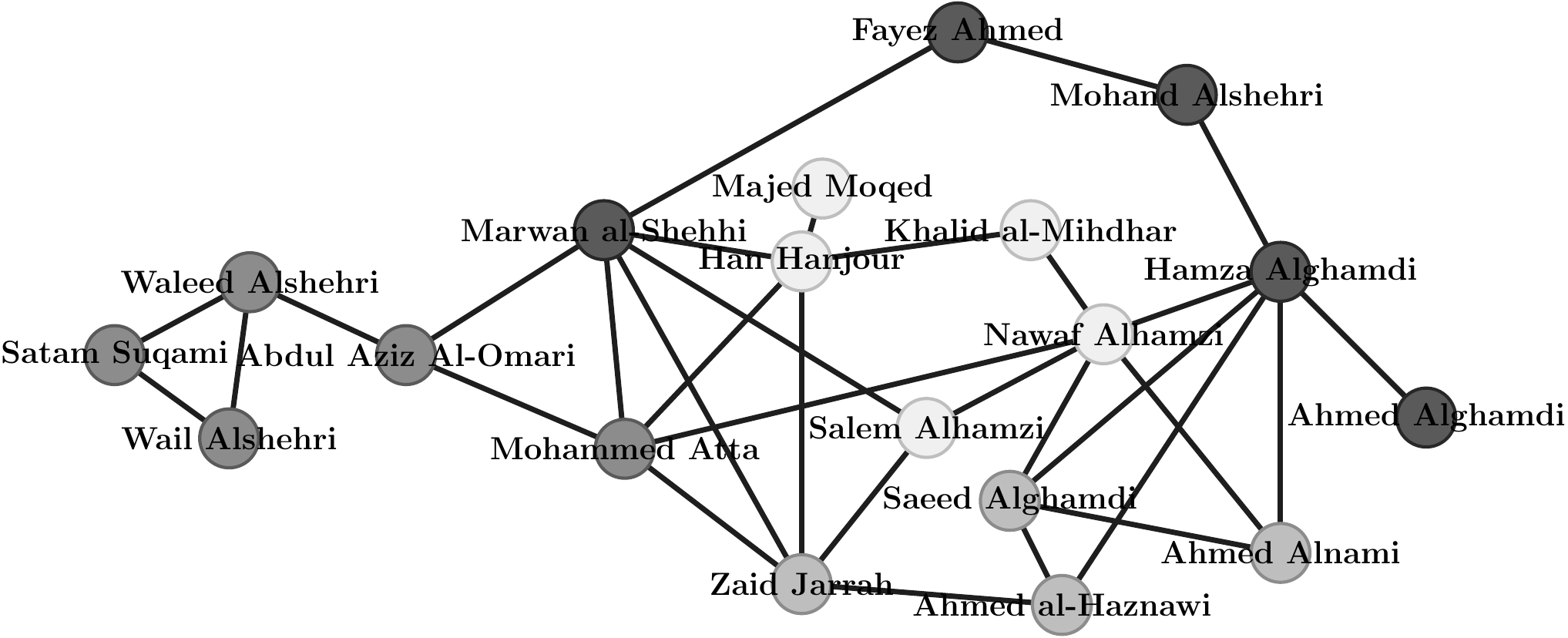}
    \caption{Graph showing the connections between the 19 hijackers that carried out the 9/11 attacks. Vertices are coloured according to the flights they were on. Note that the full network consists of 69 vertices.}
    \label{fig:911network}
\end{figure}

\citet{michalak2013computational} showed that computing the Shapley value  for the unweighted game is unfortunately $\#\P$-hard. As such, it is unlikely that an efficient algorithm for computing these values exists. On the other hand, Michalak et al. also proposed an algorithm that is slightly more efficient than the brute-force approach used by Lindelauf et al., called \emph{FasterSVCG}. Using this algorithm, the authors computed Shapley values for a larger version of the 9/11 network, with 36 vertices (corresponding to the hijackers and some key accomplices). Their approach, rather than considering all $2^{|V|}$ coalitions, considers only the \emph{connected} coalitions, of which there may be considerably fewer than $2^{|V|}$. However, in the worst case, the number of connected coalitions may still be exponential. As such, for the full 69-vertex network, running this algorithm is still infeasible.

\citet{michalak2013computational} also considered an approximation method based on random sampling and studied its performance on the 36-vertex 9/11 network. \citet{approximation} proposed a different sampling method, \emph{structured random sampling}, that aims to be more efficient than random sampling. Using this method, they computed approximate Shapley values for the 69-vertex 9/11 network. Unfortunately, neither method comes with any formal guarantees on the quality of the approximation.

While, in general, one should not expect to find an efficient algorithm for computing the exact Shapley values of these games (due to the $\#\P$-hardness), we can attempt to exploit the structure that the networks may have in order to obtain more efficient algorithms.

In this paper, we show that the Shapley value (for both the weighted and the unweighted game) can be computed efficiently on graphs of bounded \emph{treewidth}. Exploiting the fact that a graph has bounded treewidth is a celebrated and widely-used technique from theoretical computer science, and many problems are known to be solvable in polynomial time on graphs of bounded treewidth \citep{overviewtw}. We use this idea to derive a fixed-parameter tractable algorithm for computing Shapley values. Our result is not merely theoretical: we also provide an implementation and show that it can be used to compute Shapley values for graphs of practical interest.


Treewidth is often said to be a measure of how ``tree-like'' a graph is. If a graph, on a macroscopic scale, resembles a tree, then treewidth in some sense measures how much it deviates from being a tree locally. The treewidth of a graph is defined in terms of \emph{tree decompositions}. Such a decomposition is based around a tree, to each vertex of which is associated a set of vertices of the original graph in a way that respects the structure of the graph. The largest number of vertices associated in this way to any given vertex of the tree, determines the treewidth of the graph. We will make this definition more precise in the following section.

Thanks to the simple structure of trees, many $\NP$-hard problems can be solved in polynomial time on trees. Very often, such problems can also be solved in polynomial time on graphs of bounded treewidth: while the best known algorithms for such problems in general require exponential time, we can often construct an algorithm that is only exponential in the treewidth (by doing some exponential computation within each vertex of the tree, which contains a bounded number of vertices of the original graph), and then using the properties of the problem that allow it to be solved in polynomial time on trees to combine the results computed within each tree vertex to a solution for the original problem.

Using our approach, we are able to compute the exact Shapley value for the full 69-vertex network of the 9/11 attacks. 
We are also able to compute the exact Shapley values for some much larger networks, having up to a thousand vertices.

Of course, our method crucially depends on the network having bounded treewidth. Fortuitously, the network of the 9/11 attacks has a rather low treewidth of only $8$. In general one cannot expect social networks to have small treewidth: social networks often have large cliques, and the size of the largest clique forms a lower bound on the treewidth of a graph \citep[see e.g.,][for a study of tree decompositions of social networks]{socialtrees}. However, terrorist and criminal organizations are often well-served by keeping their networks sparsely connected, as this helps to avoid detection and as such one would not expect large cliques \citep{Lindelauf2009}. As another example of networks that may have low treewidth, the interaction networks in a hierarchical organization would naturally be tree-like.

Our goal is to develop an algorithm that, given a graph $G$ with $n=|V|$ and tree decomposition of $G$ of width $tw$, computes the Shapley value in time $f(tw)n^{O(1)}$, where $f$ is some exponential function and $n^{O(1)}$ a (low-degree) polynomial. As such, we hope to ``hide'' the exponential behaviour of computing the Shapley value in a function that depends only on the treewidth of the graph, and obtain an algorithm whose running time is (for graphs of bounded treewidth) polynomial in $n$.


Specifically, we show that for a graph $G$ of treewidth $tw$ and a vertex $v\in V$, 
the Shapley value of the vertex-weighted connectivity game of vertex (player) $v$ can be computed in time $2^{O(tw\log{tw})}n^4\log {n}$. Note that our algorithm for computing the Shapley value requires multiplying large ($O(n)$-bit) integers; this running time is achieved if using the $O(n \log n)$-time algorithm of \citet{multiplication-algo}. Moreover, we usually want to know 
the Shapley value for all vertices rather than for a specific vertex. Rather than running the previous algorithm $n$ times, we also show that computing the Shapley value for all vertices can be done in the same time, by reusing the intermediate results of previous computations.

For instance, the graph considered by Lindelauf et al. that represents the communications between the perpetrators of the 9/11 attacks, consists of $69$ vertices but only has treewidth $8$. While evaluating all $2^{69}$ subsets of vertices is clearly infeasible, our algorithm can compute the Shapley value in a couple of seconds thanks to the low treewidth of the graph.

Given a graph $G=(V,E)$, the algorithm of \citet{michalak2013computational} runs in time $O((|V|+|E|)|C|)$\footnote{\citet{michalak2013computational} ignore in the analysis of their running time the fact that the numbers involved in the computation of the Shapley value can get exponentially large, and thus we can no longer presume that arithmetic operations can be done in $O(1)$ time. Our running times do account for this, and are thus a factor $n\log{n}$ higher.}, where $|C|$ denotes the number of connected induced subgraphs of $G$. This algorithm, while offering a moderate improvement over the brute-force approach still requires exponential time on almost all interesting classes of graphs. We remark that there exist graphs of low treewidth that have a very large number of connected induced subgraphs (for example, the star on $n$ vertices has treewidth $1$ and more than $2^{n-1}$ induced connected subgraphs), while graphs with a small number of induced connected subgraphs also have low treewidth: a graph with at most $|C|$ induced connected subgraphs has treewidth at most $2\log{|C|}$ (in fact, pathwidth at most $2\log{|C|}$: if we fix some arbitrary vertex $v$, then there are at most $\log{|C|}$ vertices at distance exactly $r$ from $v$). While this bound is tight up to a constant factor (for instance on an $n$-vertex clique), in many instances the treewidth is much smaller than $\log{|C|}$.

Treewidth was first considered in the context of connectivity games by \citet{aziz}. They considered a game wherein the players are the edges and a winning coalition is one that spans the vertex set. They proved that computing the Banzhaf index can be done in polynomial time for a graph with bounded treewidth, but gave no experimental results and stated as an open problem whether Shapley values could be computed in a similar manner.

Recently, \citet{greco17} proposed using treewidth to compute Shapley values for matching games in graphs. In the matching game, the value of a coalition is the size of the maximum matching in the graph induced by that coalition. However, their algorithm is based on a formulation in Monadic Second Order Logic and the application of theoretical frameworks that allow counting of satisfying assignments of MSO formulas. For graphs of bounded treewidth, this yield a polynomial-time algorithm, where the degree of the polynomial may depend on the treewidth. In contrast, the degree of the polynomial in our algorithm is fixed, and only the constant factor in the running time is affected by the treewidth (i.e., we obtain a \emph{fixed-parameter tractable} algorithm). Moreover, due to the application of these algorithmic metatheorems, their algorithm is not very efficient in practice: \citet{greco17} report that, even for a graph of treewidth only $3$ with $30$ vertices, their implementation (using the MSO solver Sequoia \citep{langer13}) took nearly $9$ minutes to determine the Shapley values. We are able to process much more complex (i.e., higher treewidth) graphs with significantly more vertices in a much shorter time.

Of course, our method crucially depends on being given a tree decomposition of low width. Note that while in general, computing a minimum width tree decomposition is an $\NP$-hard problem, for many graphs of practical interest this can be done efficiently \citep[see e.g.,][for an overview of recent -- and very competitive -- implementations for computing treewidth]{pacechallenge}. 

This paper is organized as follows. We will first present some preliminaries on both graph theory and game theoretic centrality, then present the algorithm for computing the Shapley value: we first show how we can compute the Shapley value for one specific vertex (that appears in the root bag of our decomposition), then we show how a (nice) tree decomposition can be modified to quickly compute the Shapley value for all vertices (more quickly than computing it for each vertex individually). We then present an experimental evaluation of our algorithm, evaluating the performance of our network on several benchmark graphs and real-world examples of covert networks.

\section{Preliminaries}

\subsection{Graphs and Treewidth}

Let $G=(V,E)$ be an undirected graph, where $V$ is its vertex set and $E$ its edge set. To avoid confusion when dealing with multiple graphs (with different vertex sets), we may use the notation $V(G)$ to refer to the vertex set $V$ of $G$ (and similarly, $E(G)$ to refer to its edge set $E$). Given a subset $V'\subseteq V$, we denote by $G[V']$ the subgraph of $G$ induced by $V'$. We say that a vertex set $S$ separates vertex sets $A,B$ if any path from a vertex in $A$ to a vertex in $B$ must necessarily include a vertex in $S$. Where confusion is unlikely, we may write $v\in G$ instead of $v\in V(G)$ and $V'$ instead of $G[V']$. In the following, we let $n=|V(G)|$.

Given a graph $G=(V,E)$, a tree decomposition of $G$ is a tree $T$ together with for each vertex $t\in V(T)$ a subset $X_t\subseteq V(G)$ (called \emph{bag}) such that

\begin{enumerate}
	\item for all $v\in V(G)$, there is a $t\in V(T)$ such that $v\in X_t$,

	\item for all $(u,v)\in E(G)$, there is a $t\in V(T)$ such that $\{u,v\}\subseteq X_t$,
	
	\item for any $v\in V(G)$, the subset $\{t\in V(T) \mid v\in X_t\}$ induces a connected subtree of $T$.
\end{enumerate}

The \emph{width} of a tree decomposition is $\max_{t\in T} |X_t| - 1$, and the treewidth of a graph $G$ is the minimum width taken over all tree decompositions of $G$. To avoid confusion, from now on we shall refer to the vertices of $T$ as ``nodes'', and ``vertex'' shall refer exclusively to vertices of $G$.

We may designate an arbitrary node of $T$ as \emph{root} of the tree decomposition. Given a node $t\in T$, we denote by $G[t]$ the subgraph of $G$ induced by $X_t$ and the vertices in bags of nodes which are descendants of $t$ in $T$ (i.e. bags corresponding to vertices which can be reached from $t$ without going closer to the root). The following well-known lemma is an important fact, stating that the bags of a tree decomposition are separators:

\begin{lemma}[equivalent to \citet{fptbook}, Lemma 7.3]\label{lem1}
The vertices in $X_t$ separate $G[t]$ from the rest of the graph, i.e., for every edge $(u,v)\in E(G)$ for which $u\in G[t]$ and $v\not \in G[t]$, it holds that $u\in X_t$.
\end{lemma}

Any tree decomposition can be converted (in linear time) into a decomposition in \emph{nice} form, that is, each of the nodes $t\in T$ is one of four types \citep{kloks94}:

\begin{itemize}
	\item \textbf{Leaf:} $t$ is a leaf of $T$, and $|X_t|=1$.
	
	\item \textbf{Introduce:} $t$ has a single child node $t'$. $X_t\supset X_{t'} $ and $X_t$ contains exactly one vertex $v\in V(G)$ that is not in $X_{t'}$, i.e., $X_t=X_{t'}\dot{\cup} \{v\}$. We say that $v$ is \emph{introduced} in $t$.
	
	\item \textbf{Forget:} $t$ has a single child node $t'$. $X_t \subset X_{t'}$ and $X_{t'}$ contains exactly one vertex $v\in V(G)$ that is not in $X_t$, i.e., $X_{t'}=X_t\setminus \{v\}$. We say that $v$ is \emph{forgotten} in $t$.
	
	\item \textbf{Join:} $t$ has exactly two children $l,r$. Moreover, $X_l=X_r=X_t$.
\end{itemize}

If the tree decomposition is given in nice form, we can specify an algorithm simply by specifying how it processes each of these four cases. Moreover, we can assume that the size of a (nice) tree decomposition (i.e. the number of bags) is linear in $n$ \citep{kloks94}.

\subsection{Shapley Value and Game-Theoretic Centrality}

A \emph{coalitional game} consists of a set of players $N$ (the \emph{grand coalition}) together with a \emph{characteristic function} $w:2^N \to \mathbb{R}$ such that $w(\emptyset) = 0$. Given a characteristic function, the Shapley value $\Phi_i(w)$ of a player $i$ is defined as \citep{shapleyvalue}:

\begin{equation}\label{eqn:shapleyvalue}
\Phi_i(w) = \sum_{S\subseteq N\setminus\{i\}} \frac{|S|!(|N|-|S|-1)!}{|N|!} (w(S\cup \{i\}) - w(S))
\end{equation}





In this paper, we consider coalitional games where the players correspond to vertices in a graph. The connectivity game $v^{conn}$, introduced by \citet{amer}, is given by the weight function:

\[
v^{conn}(S) =
    \begin{cases}
      1 & \text{if}\ G[S] \ \text{is connected and}\ |S|>1\text{,} \\
      0 & \text{otherwise.}
    \end{cases}
\]

Note that a coalition consisting of a single player, while connected, has a value of $0$.

\citet{lindelauf13} consider vertex-weighted connectivity games as a generalization of connectivity game. They assume each vertex $i$ has a weight $w(i)$ and the corresponding
vertex-weighted connectivity game $v^{wconn2}$ is defined as follows:


\[
v^{wconn2}(S) =
    \begin{cases}
      \sum_{i\in S} w(i) & \text{if}\ G[S] \ \text{is connected,} \\ 
      0 & \text{otherwise.}
    \end{cases}
\]


In this paper, we give an algorithm for computing the Shapley value associated with $v^{wconn2}$, that, as a byproduct, also computes the Shapley value associated with $v^{conn}$.


\section{The Algorithm}

In this section, we present our algorithm for computing the Shapley values of $v^{conn}$ and $v^{wconn2}$. We begin by presenting an algorithm that computes the Shapley value for a specific vertex $v$ if a nice tree decomposition which contains $v$ as sole vertex in its root bag is given. We then show how an arbitrary (nice) tree decomposition can be modified to contain any vertex in its root bag, allowing us to evaluate the Shapley value for any vertex. Finally, we show how to avoid the extra factor $n$ that would appear in the running time if we computed the Shapley value for each vertex individually, by reusing parts of the computation.

\begin{theorem}\label{thm:main}
Given a graph $G=(V,E)$ and a nice tree decomposition $T$ of width $tw$ such that the root bag $X_r$ contains only a single vertex $v$, $\phi_v(v^{conn})$ and $\phi_v(v^{wconn2})$ can be computed in time $2^{O(tw\log{tw})}n^4\log{n}$.
\end{theorem}

Pseudocode for our algorithm is given in Listing~\ref{lst:main}, which uses procedures given in Listings~\ref{lst:leaf}, \ref{lst:introduce}, \ref{lst:forget} and \ref{lst:join}. Note that where we say we \emph{update} a value, if it has not been set previously, we initialize it to $0$.

We give the algorithm for computing $\phi_v(v^{wconn2})$; the results obtained from this algorithm can also be used to compute $\phi_v(v^{conn})$. We first show that for a given $v\in V$, the (single) value $\phi_v(v^{wconn2})$ can be computed in time $2^{O(tw\log{tw})}n^4\log{n}$. 

\begin{algorithm}[h]
\caption{Algorithm for computing the Shapley value -- Main Procedure}
\label{lst:main}
\begin{algorithmic}
	\Require Nice tree decomposition $(T,\{X_t \mid t\in T\})$ with root node $r$, such that $X_r=\{v\}$
	\Ensure Shapley value $\phi_v(v^{conn})$ (resp. $\phi_v(v^{wconn2})$).
	\Procedure{ComputeShapley}{$T,r$}
		\State \Let $r$ be the root of $T$
		\For{\Each node $t\in T$ in post-order}
		    \If{$t$ is a leaf node}
		        \State \Let $(w_t, n_t) = \textsc{Leaf}(t)$
		    \EndIf
		
    		\If{$t$ is an introduce node}
		        \State \Let $t'$ be the child of $t$
		        \State \Let $n_{t'}$, $w_{t'}$ be the previously computed tables for $t'$
		        \State \Let $(w_t, n_t) = \textsc{Introduce}(t, n_{t'}, w_{t'})$
		    \EndIf
		
		    \If{$t$ is a forget node}
		        \State \Let $t'$ be the child of $t$
		        \State \Let $n_{t'}$, $w_{t'}$ be the previously computed tables for $t'$
		        \State \Let $(w_t, n_t) = \textsc{Forget}(t, n_{t'}, w_{t'})$
		    \EndIf
		
		    \If{$t$ is an join node}
		        \State \Let $l,r$ be the children of $t$
		        \State \Let $n_l, n_r, w_l, w_r$ be the previously computed tables for $l, r$
		        \State \Let $(w_t, n_t) = \textsc{Join}(t, n_l, n_r, w_l, w_r)$
		    \EndIf
		\EndFor
		\State \Let $\phi_v(v^{conn}) = 0, \phi_v(v^{wconn2}) = 0$
		\State \Let $\phi_v(v^{conn}) = \frac{1}{|V|!} \Bigg(n_r(\{v\},\{\{v\}\}, 2) \cdot (|V|-2)! + \Sigma_{i=2}^{|V|-1} \Big(n_r(\{v\},\{\{v\}\},i + 1) - n_r(\emptyset,\emptyset,i)\Big)\cdot i! \cdot (|V|-i-1)!\Bigg)$
		\State \Let $\phi_v(v^{wconn2}) = \frac{1}{|V|!} \Sigma_{i=1}^{|V|-1} \Big(w_r(\{v\},\{\{v\}\},i + 1) - w_r(\emptyset,\emptyset,i)\Big)\cdot i! \cdot (|V|-i-1)!$
		\State \Return $(\phi_v(v^{conn}),\phi_v(v^{wconn2}))$
	\EndProcedure
\end{algorithmic}
\end{algorithm}

Using Equation~(\ref{eqn:shapleyvalue}) we obtain the following more suitable expression for computing $\phi_v(v^{wconn2})$, by splitting the summation into different terms, depending on the cardinality $k$ of $S$:

\begin{equation*}
\resizebox{0.999\hsize}{!}{$\begin{split}
\phi_v(v^{wconn2}) & = \sum_{k=0}^{|V|-1} \Bigg ( \frac{k!(|V|-k-1)!}{|V|!} \sum_{S\subseteq V\setminus\{v\},|S|=k}  (v^{wconn2}(S\cup \{v\}) - v^{wconn2}(S)) \Bigg ) \\
& = \sum_{k=0}^{|V|-1} \Bigg ( \frac{k!(|V|-k-1)!}{|V|!} \Big ( \sum_{S\subseteq V\setminus\{v\},|S|=k}  v^{wconn2}(S\cup \{v\}) - \sum_{S\subseteq V\setminus\{v\},|S|=k} v^{wconn2}(S) \Big ) \Bigg )
\end{split}$}
\end{equation*}

Let $S\subseteq V$. Since $v^{wconn2}(S)=0$ whenever $S$ induces a subgraph with more than one connected component, the problem of computing $\phi_v(v^{wconn2})$ reduces to computing, for each $k$, the total weight of \emph{connected} subsets $S\subseteq V(G)$ with $|S|=k$ and $v\in S$ (resp. $v\not \in S$). For $v^{conn}$, we simply need to \emph{count} the number of such subsets rather than compute their total weight.


As is standard for algorithms using dynamic programming on tree decompositions, for each node $t$ of the tree decomposition we consider the subgraph $G[t]$. For each such subgraph, we consider the subsets (coalitions) $S\subseteq G[t]$.

Recall that if $S\subseteq V(G)$ is not connected, it (by definition) does not contribute to the Shapley value. Call a subset $S\subseteq V(G[t])$ \emph{good} if the subgraph $G[S]$ induced by $S$ is connected or every connected component of $G[S]$ has non-empty intersection with $X_t$. By definition, the empty set is good.

\FloatBarrier

Our algorithm works by, for each $t\in T$, considering all good subsets $S\subseteq V(G[t])$. The following Lemma shows that subsets that are not good do not count towards the Shapley value of the game, and thus we can safely disregard them.

\begin{lemma}\label{lem:conninduced}
Let $S\subseteq V(G)$ induce a connected subgraph of $G$ and let $t\in T$. Then
$S \cap V(G[t])$ is a good subset of $G[t]$.
\end{lemma}

\begin{proof}
By contradiction. Suppose $S\cap V(G[t])$ is not connected. Then some component of $S\cap V(G[t])$ has an empty intersection with $X_t$. Then $S$ can not be connected, since by Lemma \ref{lem1} $X_t$ separates $G[t]$ from the rest of the graph.
\end{proof}


Of course, there can still be exponentially many good subsets. The key to our algorithm is that for each such subset $S$, we do not need to know exactly how the subset is made up: if we know how the subset $S$ behaves within $X_t$, we know how it interacts with the rest of the graph (outside of $G[t]$), since $X_t$ is a separator. Subsets which behave similarly within $X_t$ can be grouped together, thus speeding up the computation. We classify the subsets into groups depending on their interaction with the rest of the graph. Specifically, each subset $S\subseteq G[t]$ has a \emph{characteristic} (w.r.t. $G[t]$) that consists of

\begin{itemize}
\item the intersection $R=S\cap X_t$,

\item an equivalence relation $\sim$ on $S\cap X_t$ such that $a\sim b$ if and only if $a$ and $b$ are in the same connected component of the subgraph induced by $S$,

\item the cardinality of $S$, $k=|S|$.
\end{itemize}

For the equivalence relation $\sim$, note that each element can be in one of at most $tw+1$ equivalence classes, a trivial upper bound on the number of such relation is $(tw+1)^{tw+1}$. The number of subsets $R$ is at most $2^{tw+1}$; this is dominated by the number of equivalence relations. $k$ can take values in the range $0,\ldots,n$.  Thus, the total number of distinct characteristics is $2^{O(tw\log{tw})}n$. For every node $t\in T$ and each characteristic $(R, \sim, k)$, we will compute

\begin{itemize}
\item $n_t(R, \sim, k)$: the number of good subsets $S\subseteq G[t]$ with characteristic $(R, \sim, k)$,
\item $w_t(R, \sim, k)$: the total weight of all good subsets $S\subseteq G[t]$ with characteristic $(R, \sim, k)$.
\end{itemize}

Note that the weight of a subset $S\subseteq G[t]$ is simply the sum of the vertex weights, i.e. $\sum_{v\in S} w(v)$.

Note that if $r$ is the root of $T$, and $X_r=\{v\}$. Then $w_r(\{v\},\{v\}, k)$ is exactly the total weight of connected subsets $S\subseteq V(G)$ with $|S|=k$ and $v\in S$, whereas $w_r(\emptyset,\emptyset, k)$ is the total weight of connected subsets of size $k$ not including $v$. This gives us exactly the information we need to compute $\phi_v(v^{wconn2})$.

\begin{figure}
     \centering
     \hfill
     \subfloat[t][Example of a graph] {
         \includegraphics[width=0.45\textwidth,valign=c]{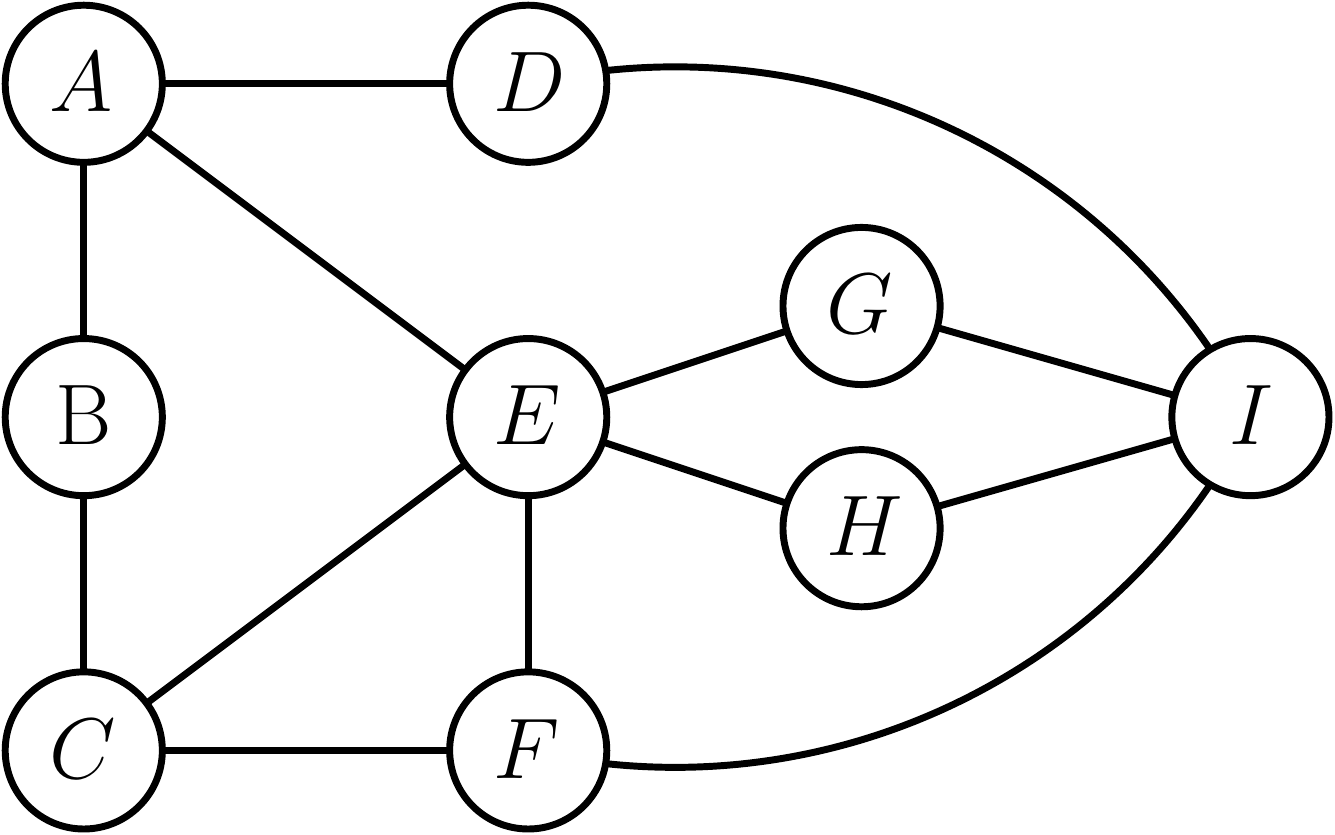}
         \vphantom{\includegraphics[width=0.45\textwidth,valign=c]{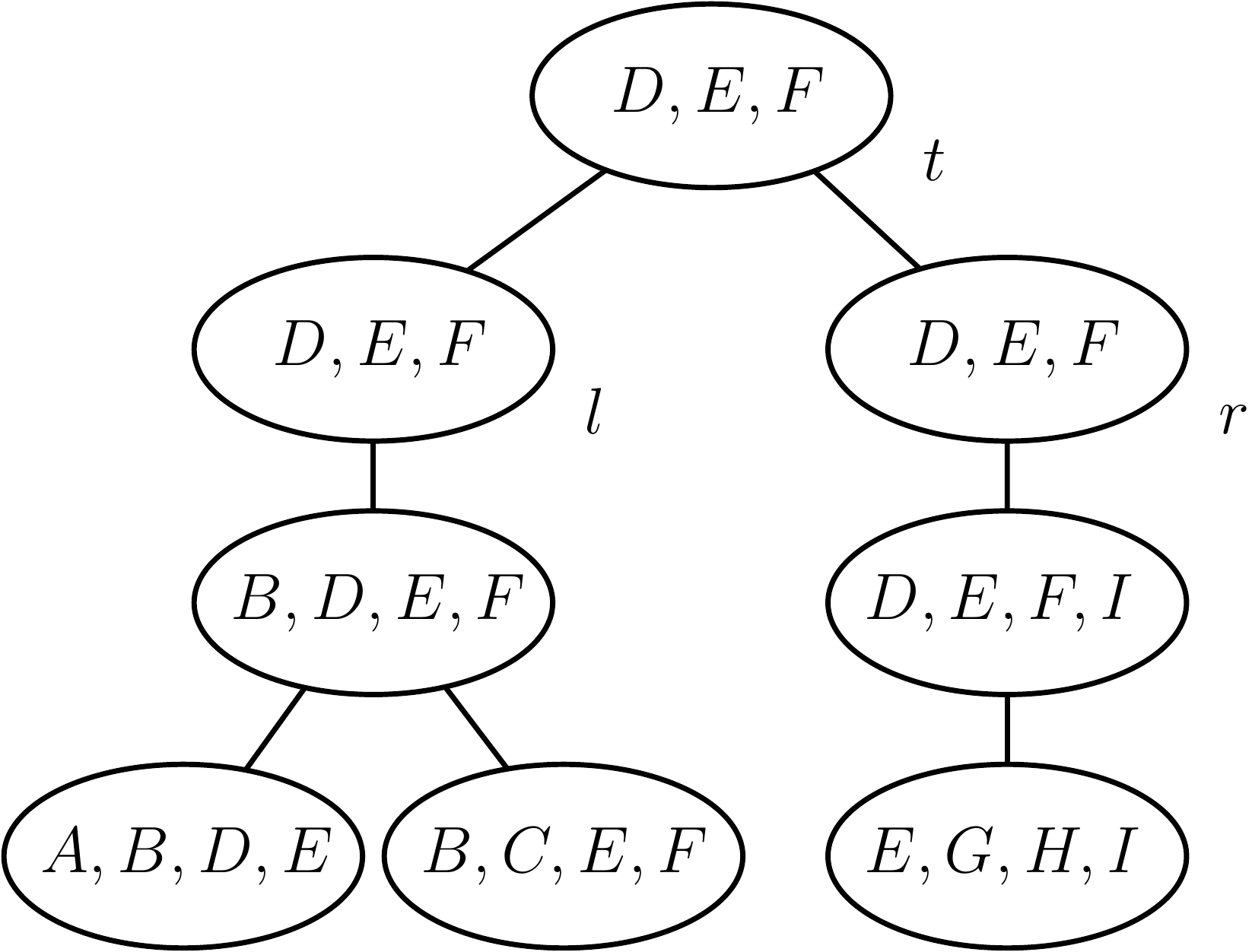}}
     }
     \hfill
     \subfloat[t][Example of a tree decomposition] {
         \includegraphics[width=0.45\textwidth,valign=c]{exampletd.pdf}
         \vphantom{\includegraphics[width=0.45\textwidth,valign=c]{exampletd.pdf}}
     }
     \hfill\null
     \caption{(a) Example of a 9-vertex graph. (b) Tree decomposition for the graph in (a); note that the decomposition given is neither of optimal width nor nice.}
     \label{fig:dp_example}
\end{figure}

The following example illustrates the characteristics $(R, \sim, k)$.
\begin{example}\label{example:td}

\sloppy Consider the graph shown in Figure~\ref{fig:dp_example}a and the tree decomposition shown in Figure~\ref{fig:dp_example}b. The induced subgraph $G[r]$ associated with node $r$ consists of vertices $D,E,F,G,H,I$. The subset $\{G,D\}$ is not good because the connected component $\{G\}$ has an empty intersection with $X_r=\{D,E,F\}$. The subset $\{D,E,G,H\}$ is good, and has characteristic $(\{D,E\}, \{\{D\},\{E\}\}, 4)$ and is the only subset having this characteristic, thus $n_l(\{D,E\}, \{\{D\},\{E\}\}, 4)=1$. The subset $\{D,E,G,I\}$ is also good, and has characteristic $(\{D,E\}, \{\{D,E\}\}, 4)$. Since the subset $\{D,E,H,I\}$ has the same characteristic, we have that $n_r(\{D,E\}, \{\{D,E\}\}, 4)=2$.
$\hfill\Box$\\
\end{example}

For every node $t\in T$ we compute $n_t(R, \sim, k)$ and $w_t(R, \sim, k)$ for each characteristic $(R, \sim, k)$ in a bottom-up fashion. We start at the leaf vertices, and then work our way up the root of the tree. We handle each of the cases as follows:

\emph{Leaf.} If $t\in T$ is a leaf node, then $X_t=\{v\}$ for some $v\in V$. Since $G[t]$ is a singleton vertex, $t$ has exactly two characteristics $c_1=(\emptyset, \emptyset, 0)$ and $c_2=(\{v\}, \{(v,v)\}, 1)$ (corresponding to the only two subsets $S \subset G[t]$, which are the empty set and the singleton $\{v\}$). It is easy to see that $n_t(c_1)=1, w_t(c_1)=0$ and $n_t(c_2)=1, w_t(c_1)=w(v)$. Pseudocode for the Leaf procedure is given in Listing~\ref{lst:leaf}.

\begin{algorithm}
\caption{Algorithm for computing the Shapley value -- Leaf Procedure}
\label{lst:leaf}
\begin{algorithmic}
	\Require Leaf node $t$ of a tree decomposition.
	\Ensure Dynamic programming tables $n_t$ and $w_t$.
	\Procedure{Leaf}{$t$}
		\State \Let $v\in V(G)$ be the vertex such that $X_t=\{v\}$
		\State \Let $n_t(\emptyset,\emptyset,0)=1$, $w_t(\emptyset,\emptyset,0)=0$ and 
$n_t(\{v\},\{\{v\}\},1)=1$, $w_t(\{v\},\{\{v\}\},1)=w(v)$
		\State \Return $(n_t,w_t)$
	\EndProcedure
\end{algorithmic}
\end{algorithm}

\FloatBarrier

\emph{Introduce.} If $t\in T$ is an introduce node, it has a single child $t'\in T$ and $X_t = X_{t'} \dot{\cup} \{v\}$ for some $v\in V(G)$. Every characteristic $(R,\sim,k)$ (w.r.t. $G[{t'}]$) corresponds to $n_{t'}(R,\sim,k)$ distinct subsets of $G[{t'}]$, and we may extend these subsets $S\subseteq G[t']$ to subsets of $G[t]$ by either adding the introduced vertex $v$ or not. Thus, the $n_{t'}(R,\sim,k)$ subsets of $G[{t'}]$ give rise to

\begin{enumerate}
	\item when not adding $v$, to $n_{t'}(R,\sim,k)$ good subsets of $G[t]$ with characteristic $(R,\sim,k)$ and total weight $w_t(R,\sim,k)$, and
	
	\item when adding $v$, if $k = 0$ or $R\not = \emptyset$, to $n_{t'}(R,\sim,k)$ good subsets of $G[t]$ with characteristic $(R\cup\{v\},\sim',k+1)$ and total weight $w_{t'}(R,\sim,k) + n_{t'}(R,\sim,k) \cdot w(v)$,
	
	\item when adding $v$, if $k\not = 0$ and $R=\emptyset$, then $S\cup \{v\}$ has at least two connected components, (at least) one of which does not intersect $X_t$, so it is not a good subset,
\end{enumerate}

where $\sim'$ is the relation obtained as the transitive closure of $\sim \cup \{(v,v)\}\cup \{(v,x) \mid x\in R, (v,x)\in E(G)\}$.

Note that two distinct characteristics $(R,\sim,k)$ and $(R',\sim',k')$ with $R=R'$ and $k=k'$ (but $\sim\not = \sim'$) may give rise (upon addition of the vertex $v$) to $n_{t'}(R,\sim,k)+n_{t'}(R',\sim',k')$ subsets with the same characteristic $(R\cup \{v\}, \sim'', k+1)$ with total weight $w_{t'}(R,\sim,k) + w_{t'}(R',\sim',k') + (n_{t'}(R,\sim,k) + n_{t'}(R',\sim',k')) \cdot w(v)$. Therefore, we can compute $n_t(R\cup\{v\},\sim',k+1)$ (and similarly, $w_t(R\cup\{v\},\sim',k+1)$), by taking the sum of $n_{t'}(R,\sim,k)$ over all $\sim$ such that $\sim'$ is the transitive closure of $\sim \cup \{(v,v)\}\cup \{(v,x) \mid x\in R, (v,x)\in E(G)\}$. Pseudocode for the Introduce procedure that illustrates this summation is given in Listing~\ref{lst:introduce}.

\begin{algorithm}
\caption{Algorithm for computing the Shapley value -- Introduce Procedure}
\label{lst:introduce}
\begin{algorithmic}
	\Require Introduce node $t$ of a tree decomposition; dynamic programming tables $n_{t'}$ and $w_{t'}$ for child node $t'$.
	\Ensure Dynamic programming tables $n_t$ and $w_t$.
	\Procedure{Introduce}{$t,n_{t'},w_{t'}$}
		\State \Let $v\in V(G)$ be the vertex such that $X_t=\{v\} \dot {\cup}X_{t'}$
		\For{\Each $(R,\sim,k)$ in the domain of $n_{t'}$}
			\State \Update $n_t(R,\sim,k)=n_t(R,\sim,k) + n_{t'}(R,\sim,k)$
			\State \Update $w_t(R,\sim,k)=w_t(R,\sim,k) + w_{t'}(R,\sim,k)$
			\If{$R\not = \emptyset$ \textbf{or} $k=0$}
				\State \Let $\sim' = \sim$
				\For{\Each $u\in R$}
					\If{$(u,v)\in E(G)$}
						\State \textbf{union} the partition classes of $u$ and $v$ in $\sim'$
					\EndIf
				\EndFor
				\State \Update $n_t(R\cup\{v\},\sim',k+1)=n_t(R\cup\{v\},\sim',k+1) + n_{t'}(R,\sim,k)$
				\State \Update $w_t(R\cup\{v\},\sim',k+1)=w_t(R\cup\{v\},\sim',k+1) + w_{t'}(R,\sim,k) + n_{t'}(R,\sim,k) \cdot w(v)$
			\EndIf
		\EndFor
		\State \Return $(n_t,w_t)$
	\EndProcedure
\end{algorithmic}
\end{algorithm}

The following lemma (c.f. Lemma \ref{lem:conninduced}) ensures the correctness of the introduce step:

\begin{lemma}
Let $t\in T$ and suppose that $t$ is an introduce node with child $t'\in T$. Suppose $S\subseteq G[t]$ is a good subset. Then $S\cap G[{t'}]$ is a good subset of $G[{t'}]$.
\end{lemma}

\begin{proof}
Suppose that $S\cap G[{t'}]$ is not connected, and some connected component $C$ of $S\cap G[{t'}]$ has an empty intersection with $X_{t'}$. Suppose the introduced vertex is $v$. Then $v$ must be adjacent to some vertex of $C$, but this is impossible since $C\cap X_{t'}=\emptyset$ and $v$ is not incident to $G[{t'}]\setminus X_{t'}$.

This ensures that we count each good subset $S\subseteq G[t]$ at least once. The at most once statement follows from the fact that $S\cap G[{t'}]$ corresponds to a unique characteristic w.r.t $G[t]$.
\end{proof}

\FloatBarrier

\begin{algorithm}[h]
\caption{Algorithm for computing the Shapley value -- Forget Procedure}
\label{lst:forget}
\begin{algorithmic}
	\Require Forget node $t$ of a tree decomposition; dynamic programming tables $n_{t'}$ and $w_{t'}$ for child node $t'$.
	\Ensure Dynamic programming tables $n_t$ and $w_t$.
	\Procedure{Forget}{$t,n_{t'},w_{t'}$}
		\State \Let $v\in V(G)$ be the vertex such that $X_t=X_{t'} \setminus \{v\}$
		\For{\Each $(R,\sim,k)$ in the domain of $n_{t'}$}
			\If{$v\not\in R$}
				\State \textbf{update} $n_t(R,\sim,k)=n_t(R,\sim,k) + n_{t'}(R,\sim,k)$
				\State \Update $w_t(R,\sim,k)=w_t(R,\sim,k) + w_{t'}(R,\sim,k)$
			\ElsIf{$|R|=1$ \textbf{or} $v$ is not a singleton in $\sim$}
				\State \Let $\sim'$ be the restriction of $\sim$ to $X_t$
				\State \Update $n_t(R\setminus\{v\},\sim',k)=n_t(R\setminus\{v\},\sim',k) + n_{t'}(R,\sim,k)$
				\State \Update $w_t(R\setminus\{v\},\sim',k)=w_t(R\setminus\{v\},\sim',k) + w_{t'}(R,\sim,k)$
			\EndIf
		\EndFor
		\State \Return $(n_t,w_t)$
	\EndProcedure
\end{algorithmic}
\end{algorithm}

\emph{Forget.} If $t\in T$ is a forget node, it has a child ${t'}\in T$ such that $X_{t'} \dot{\cup} \{v\} = X_t$ for some $v\in V(G)$. If for characteristic $(R,\sim,k)$ (w.r.t. $G[{t'}]$), $v\not\in R$, then $(R,\sim,k)$ is also a characteristic w.r.t. $G[t]$. If $v\in R$, then there are three cases:

\begin{enumerate}
\item $R=\{v\}$. Then we obtain $n_{t'}(R,\sim,k)$ good subsets of $G[t]$ with characteristic $\{\emptyset,\emptyset,k\}$ and total weight $w_t(R,\sim,k)$,

\item $R\not =\{v\}$ and $\{(v,v)\}\in \sim$. Then none of the $n_{t'}(R,\sim,k)$ good subsets of $G[{t'}]$ are good for $G[t]$, since the connected component containing $v$ does not intersect $X_t$, and there is some other connected component that does intersect $X_t$.

\item Otherwise, we obtain $n_{t'}(R,\sim,k)$ good subsets of $G[t]$ with characteristic $(R\cap X_t, \sim', k)$,
\end{enumerate}

where $\sim'$ is the relation obtained by projecting the relation $\sim$ on $R$ to $R\cap X_t$ (i.e., $\sim' = \sim \cap \{(u,v) \mid u,v\in X_t\}$).

As with the introduce procedure, subsets with a given characteristic wrt. $G[t]$ may correspond to subsets with different characteristics for $G[t']$, so to compute the table entries wrt. $G[t]$ we must once again take the sum of relevant table entries wrt. $G[t']$.  Pseudocode for the Forget procedure is given in Listing~\ref{lst:forget}. The correctness follows from the following Lemma:

\begin{lemma}
Let $t\in T$ and suppose that $t$ is a forget vertex with child ${t'}\in T$. Suppose $S\subseteq G[t]$ is a good subset. Then $S$ is a good subset of $G[{t'}]$.
\end{lemma}

\begin{proof}
If $S$ is not connected, then $S$ has non-empty intersection with $X_t$. Since $X_t\subset X_{t'}$, $S$ also has non-empty intersection with $X_{t'}$.
\end{proof}



\FloatBarrier

\emph{Join.} If $t\in T$ is a join node, then it has two children $l,r$ such that $X_l=X_r=X_t$. Suppose that $(R_l, \sim_l, k_l)$ is a characteristic of $l$ and $(R_r, \sim_r, k_r)$ is a characteristic of $r$ and suppose that $R_l=R_r$. Then there are $n_l(R_l, \sim_l, k_l) \cdot n_r(R_r, \sim_r, k_r)$ subsets with characteristic $(R_l, \sim', k_l + k_r - |R_l|)$ and total weight $n_l(R_l, \sim_l, k_l) \cdot w_r(R_r, \sim_r, k_r) + n_r(R_r, \sim_r, k_r) \cdot w_l(R_l, \sim_l, k_l) - n_l(R_l, \sim_l, k_l) \cdot n_r(R_r, \sim_r, k_r) \cdot (\Sigma_{v\in R_l} w(v))$, where $\sim$ is the transitive closure of $\sim_l\cup \sim_r$. Pseudocode for the Join procedure is given in Listing~\ref{lst:join}.

Similarly to the introduce and forget cases, multiple distinct characteristics for $l,r$ may, when combined, correspond to the same characteristic for $t$; we should again take the sum over these characteristics. The correctness follows from the following Lemma:

\begin{lemma}
Let $t\in T$ and suppose that $t$ is a join node with children $l,r\in T$. Suppose $S\subseteq G[t]$ is a good subset. Then $S\cap V(G[l])$ (resp. $S\cap V(G[r])$) is a good subset of $G[l]$ (resp. $G[r]$).
\end{lemma}

\begin{proof}
By contradiction. We show the case for the left child, the case for the right child is symmetric.

If $S\subseteq V(G[l])$ then the lemma follows automatically. Therefore, assume there exists $v\in S$ such that $v\notin V(G[l])$. In particular, this means that $v\in G[r]\setminus X_r$.

Suppose $S\cap V(G[l])$ is not connected and has a connected component $C$ with empty intersection with $X_l$. Since none of the vertices of $S\cap V(G[r])$ are incident to $C$, $C$ is still a maximal connected component of $S$, but $S$ has at least one other connected component (since $S\cap V(G[l])$ is not connected) and so is not connected, and $C$ has empty intersection with $X_l=X_t$.
\end{proof}

\begin{algorithm}
\caption{Algorithm for computing the Shapley value -- Join Procedure}
\label{lst:join}
\begin{algorithmic}
	\Require Join node $t$ of a tree decomposition; dynamic programming tables $n_l$, $n_r$, $w_l$ and $w_r$ for child nodes $l,r$.
	\Ensure Dynamic programming tables $n_t$ and $w_t$.
	\Procedure{Join}{$t,n_l,n_r,w_l,w_r$}
		\For{\Each $(R,\sim,k)$ in the domain of $n_l$}
			\For{\Each $(R',\sim',k')$ in the domain of $n_r$ such that $R=R'$}
				\State \Let $\sim''$ be the transitive closure of $\sim$, $\sim'$
				\State \Update $n_t(R,\sim'',k+k'-|R|)=n_t(R,\sim'',k+k'-|R|) + n_l(R,\sim,k) \cdot n_r(R,\sim',k')$
				\State \Update $w_t(R,\sim'',k+k'-|R|)=w_t(R,\sim'',k+k'-|R|) + n_l(R,\sim,k) \cdot w_r(R,\sim',k') + n_r(R,\sim',k') \cdot w_l(R,\sim,k) - n_l(R,\sim,k)\cdot n_r(R,\sim',k')\cdot \Sigma_{v\in R} w(v)$
			\EndFor
		\EndFor
		\State \Return $(n_t,w_t)$
	\EndProcedure
\end{algorithmic}
\end{algorithm}

As it is the most complicated procedure, we also give an example illustrating the join procedure.

\begin{example}
\sloppy Consider Figure~\ref{fig:dp_example} of Example \ref{example:td}. Consider the join node $t$ and its
children $l,r$. There are $13$ subsets\footnote{In this example, we refer exclusively to \emph{good} subsets.} (of $G[t]$) with characteristic $(\{D,E\},\{\{D,E\}\}, 5)$. This can be seen as follows, there are:
\begin{itemize}
\item 1 subset of $G[l]$ with characteristic $(\{D,E\},\{\{D,E\}\}, 5)$ times 1 subset of $G[r]$ with characteristic $(\{D,E\},\{\{D\},\{E\}\}, 2)$,
\item 2 subsets of $G[l]$ with characteristic $(\{D,E\},\{\{D,E\}\}, 4)$ times 3 subsets of $G[r]$ with characteristic $(\{D,E\},\{\{D\},\{E\}\}, 3)$,
\item 1 subset of $G[l]$ with characteristic $(\{D,E\},\{\{D,E\}\}, 3)$ times 1 subset of $G[r]$ with characteristic $(\{D,E\},\{\{D\},\{E\}\}, 4)$,
\item 1 subset of $G[l]$ with characteristic $(\{D,E\},\{\{D,E\}\}, 3)$ times 2 subsets of $G[r]$ with characteristic $(\{D,E\},\{\{D,E\}\}, 4)$,
\item 1 subset of $G[l]$ with characteristic $(\{D,E\},\{\{D\},\{E\}\}, 3)$ times 2 subsets of $G[r]$ with characteristic $(\{D,E,\{\{D,E\}\}, 4)$,
\item 1 subset of $G[l]$ with characteristic $(\{D,E\},\{\{D\},\{E\}\}, 2)$ times 1 subset of $G[r]$ with characteristic $(\{D,E,\{\{D,E\}\}, 5)$,
\end{itemize}
and we have that $1\times 1 + 2 \times 3 + 1 \times 1 + 1 \times 2 + 1 \times 2 + 1 \times 1 = 13$.
$\hfill\Box$\\
\end{example}

\FloatBarrier

By processing the vertices of the tree decomposition in a bottom-up fashion, we can compute the values $n_r(R,\sim,k)$ and $w_r(R,\sim,k)$ for all characteristics $(R,\sim,k)$ of the root node $r$. As we have seen before, knowing these values is sufficient to compute the Shapley value of vertex $v$.  Now, we are ready to provide the proof of Theorem~\ref{thm:main}.


\begin{proof}[Proof of Theorem~\ref{thm:main}]
We assume we are given a nice tree decomposition of $G$ (which we may assume has $O(n)$ nodes). For each node, there are $2^{O(tw \log tw)} n$ characteristics. To compute the values for one characteristic requires considering (in the worst case, which is the join node) $2^{O(tw \log tw)} n^2$ pairs of characteristics for the child nodes. For each such pair, we perform a constant number of multiplications of $n$-bit integers, taking $n\log{n}$ time. The dynamic programming table for one node of the tree decomposition takes up $2^{\Theta(tw \log tw)} n^2$ space, but at any given time we only need to keep $O(\log n)$ of them in memory.
\end{proof}

Of course, this only allows us to evaluate the Shapley value for a \emph{single} vertex $v$, under the assumption that for the root bag $r$, $X_r=\{v\}$ (i.e., $v$ is the only vertex in the root bag). To compute the Shapley value for \emph{all} vertices, we perform the following operations, starting from a nice tree decomposition:

\begin{itemize}
	\item For every join node $t$, we create a new node $t'$ with $X_{t'}=X_{t}$. $t'$ is made the parent of $t$, and the original parent of $t$ becomes the parent of $t'$. In case $t$ was the root, $t'$ becomes the root. Note that $t'$ is neither a join, introduce, forget, or leaf node, however, the dynamic programming tables for $t'$ are simply equal to those for $t$ (we shall from now on, refer to nodes such as $t'$ as \emph{nochange} nodes).
	
	\item For every vertex $v\in V(G)$, we pick a node of the tree decomposition $t$ such that $v\in V_t$. We create a copy $t'$ of $t$, which is made the parent of $t$, and the original parent of $t$ becomes the parent of $t'$. In case $t$ was the root, $t'$ becomes the root. Next, we create another copy $t''$ of $t'$. $t'$ is made the parent of $t''$ (making $t'$ into a join node). We then create a series of introduce nodes, starting from $t''$, such that eventually we end up with a leaf node $u$, whose bag contains only $v$. If we now reroot our tree decomposition so that the root becomes $u$, thanks to the previous transformation, every join node remains a join node -- the roles of introduce, forget and nochange nodes can become interchanged.
\end{itemize}

\begin{figure}[t]
     \centering
     \hfill
     \subfloat[t][Nice tree decomposition] {
         \includegraphics[scale=0.35,valign=c]{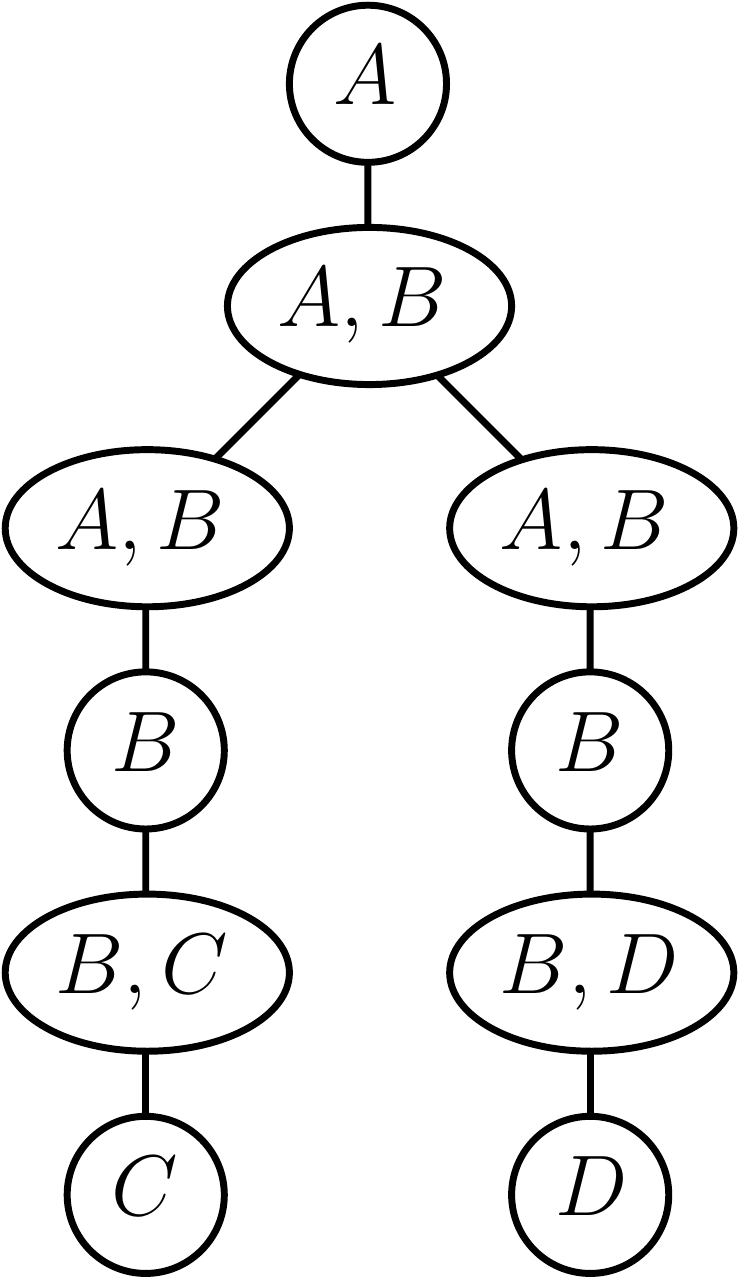}
         \vphantom{\includegraphics[scale=0.35,valign=c]{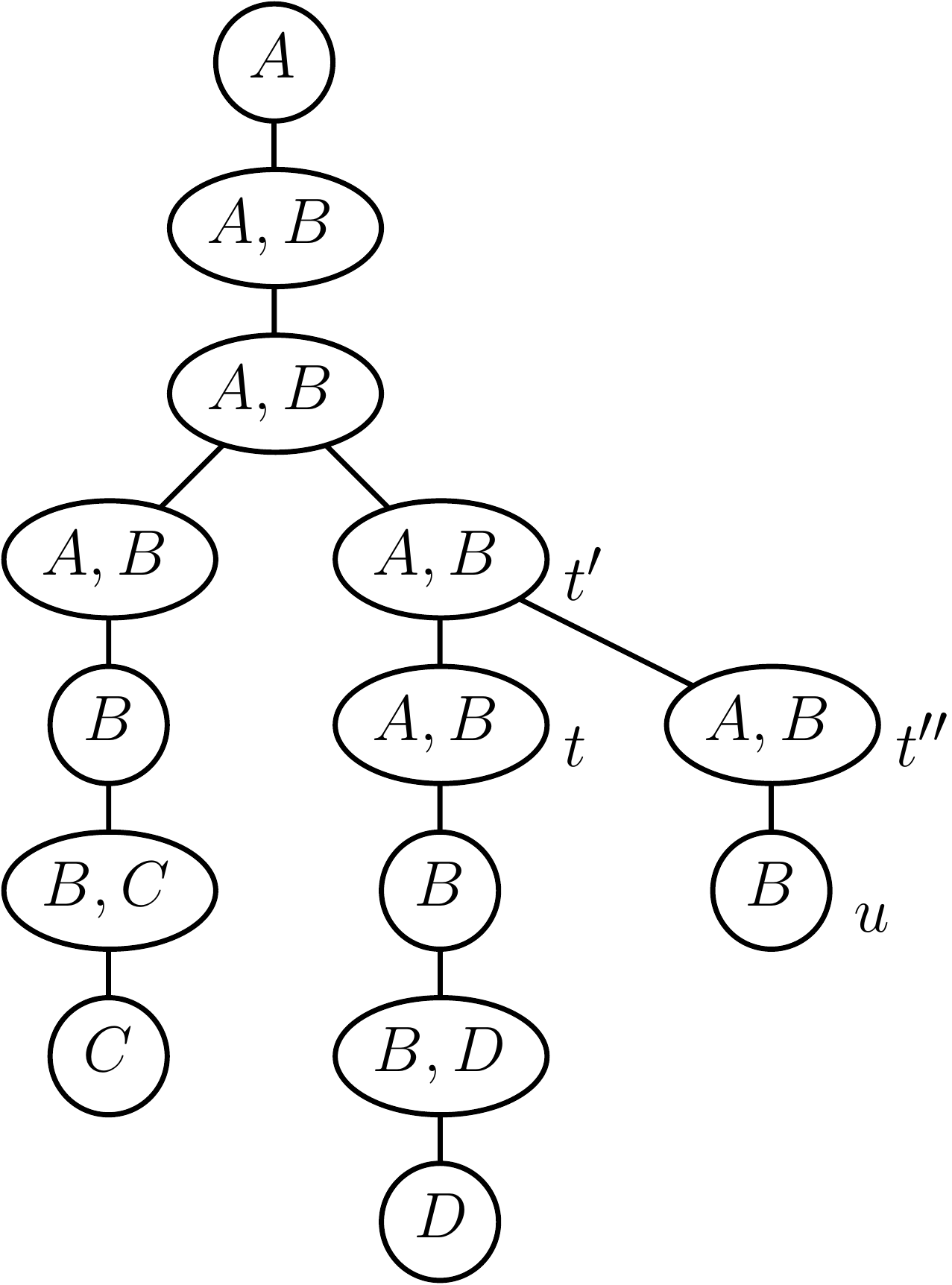}}
     }
     \hfill
     \subfloat[t][With nochange node added] {
         \includegraphics[scale=0.35,valign=c]{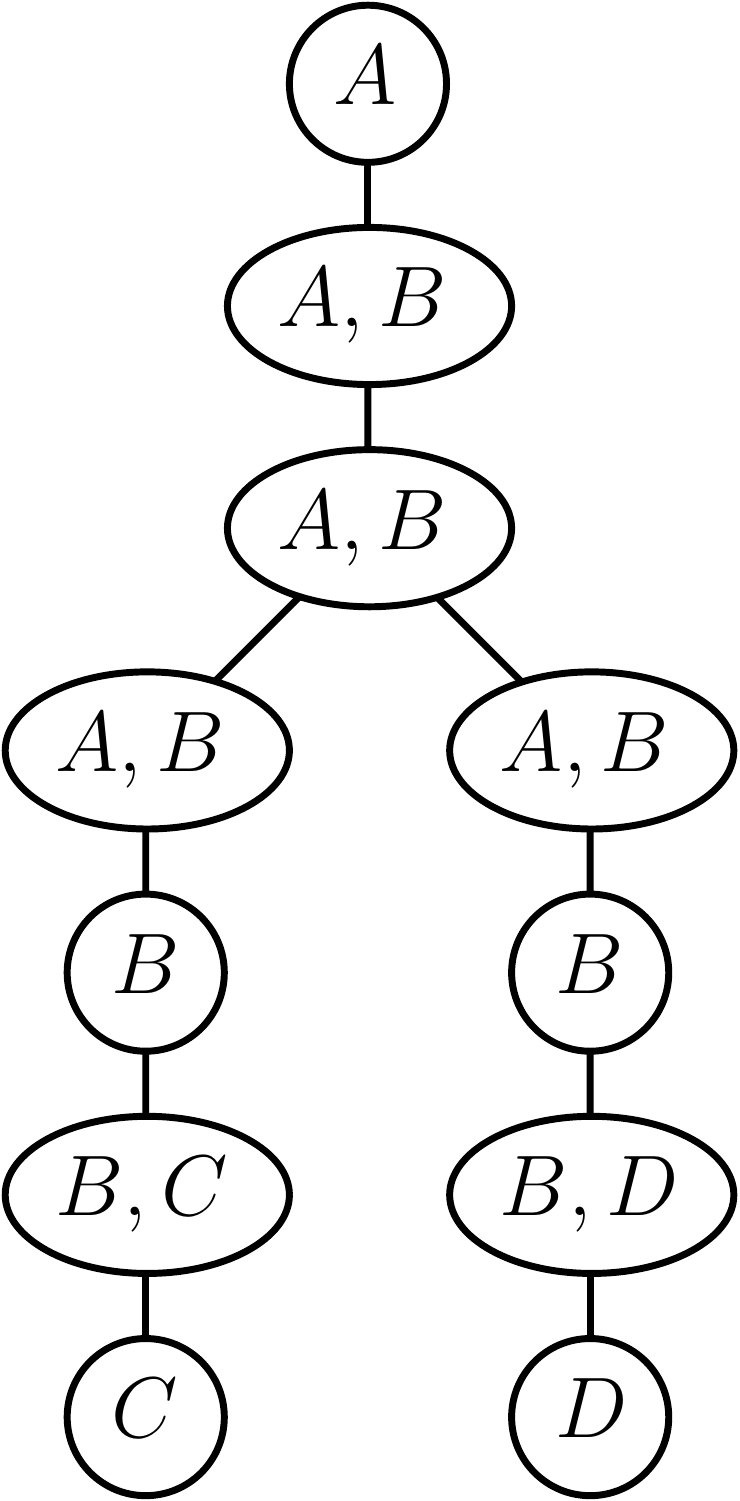}
         \vphantom{\includegraphics[scale=0.35,valign=c]{extranode.pdf}}
     }
     \hfill
          \subfloat[t][Additional nodes for vertex $B$] {
         \includegraphics[scale=0.35,valign=c]{extranode.pdf}
     }
     \hfill\null
     \caption{(a) A (nice) tree decomposition. (b) A nochange node is added before the join bag ${A,B}$. (c) Extra nodes $t'$ and $t''$ are added to enable the creation of a leaf bag containing vertex $B$, which can be used to re-root the decomposition.}
     \label{fig:nochange_example}
\end{figure}

The following example illustrates the two operations described above.
\begin{example}
Figure~\ref{fig:nochange_example} shows an example of this process. Starting from a nice tree decomposition (Figure~\ref{fig:nochange_example}a) a nochange node is added before the join bag ${A,B}$ (Figure~\ref{fig:nochange_example}b). To create a leaf bag for vertex $B$, we pick a bag $t$ containing it (in this example the right child of the join bag), insert a node $t'$ which becomes the parent of $t$, create an additional child (of $t'$) $t''$ (thus making $t'$ into a join node), then add a leaf bag $u$ (containing only $B$) as child of $t''$ (making $t''$ into a forget node).

This process can be repeated until for each vertex $v\in V(G)$ there exists a leaf bag containing it. Note that in the example the tree decomposition is rooted at $A$, but it can also be viewed as being rooted at $u$ (or any other leaf node); this turns $t''$ from a forget node into a nochange node, $t'$ remains a join node, while the nochange node ${A,B}$ (currently a child of the root node $A$) becomes an introduce node (introducing $B$ to the leaf node containing $A$).
$\hfill\Box$\\
\end{example}

Thus, we now have a tree decomposition that can be rerooted such that any vertex $v$ becomes the sole vertex in the root bag. However, this only gives a $2^{\Theta(tw \log tw)} n^5 \log{n}$-time algorithm for computing the Shapley values for all the vertices in a given graph, since this would require running the algorithm separately for each root vertex. However, there is a lot of overlap in these computations, as the dynamic programming tables for each subtree may be computed multiple times. By memoizing a table when it is computed \citep[similar to belief propagation in Bayesian Networks, see e.g.][]{pearl}, we thus obtain a $2^{O(tw \log tw)} n^4 \log{n}$-time algorithm using $2^{\Theta(tw \log tw)} n^3$ space:

\begin{theorem}
Given a graph $G$ of treewidth at most $tw$, the Shapley value of all vertices $v\in V(G)$ can be computed in time $2^{\Theta(tw \log tw)} n^4 \log{n}$ and space $2^{\Theta(tw \log tw)} n^3$.
\end{theorem}

\section{Computational Experiments}

In this section, we experimentally evaluate our algorithm. We test it on several real-world (covert) social networks and also on several (artificial) benchmark graphs. We show that our algorithm can compute the Shapley value for these networks in a reasonable amount of time.

We tested our algorithm using the following covert networks found in the literature:

\begin{itemize}
\item A network of 69 of individuals involved in 9/11 attacks (9-11), where edges represent some kind of tie (such as cooperating in an attack, financial transactions or having trained together) \citep{krebs}.

\item A network of 77 Islamic State members in Europe (ise-extended), where edges represent some kind of tie (such as cooperating in an attack, being related or being present in the same location) \citep{ise-extended}.

\item A network of 293 drug users (drugnet), where edges represent acquaintanceships \citep{drugnet}.

\item A network of 36 Montreal gangs (montreal), where edges represent ties between gangs \citep{montreal}.

\item A network of 67 members of Italian gangs (italian), where an edge represents joint membership of a gang \citep{italian}.
\end{itemize}

We also tested our algorithm on several benchmark graphs from the 2018 PACE challenge \citep{pace2018}. These graphs are not social networks but are intended to demonstrate the capabilities of our algorithms on graphs with a range of treewidth values and vertex counts.

For each network, we considered the largest connected component. Each of these networks has relatively low treewidth. The Islamic State network has the highest treewidth ($13$), while the Italian Gang network is very sparse (treewidth $3$). We also considered using the Noordins top terrorist network \citep{disrupting}. However, as this 79-vertex network has treewidth at least $19$, applying our techniques is not feasible.

Our implementation simultaneously computes the value of both $v^{conn}$ and $v^{wconn2}$ (we set all weights to $1$ for these experiments, resulting in a game where the value of a connected coalition is equal to its size).

Table~\ref{tab:performance} reports computational performance results on these benchmark graphs. Our implementation uses the .NET BigInteger library, which performs multiplications in $\Theta(n^2)$ time using a method similar to grid multiplication. While there are several asymptotically faster methods for multiplication, and we experimented with several such implementations, none of these resulted in a significant speed up for the graphs considered. The time reported is that for computing the Shapley values of \emph{all} vertices in the graph, using the method that stores all intermediate tables to achieve a $2^{O(tw \log{tw})} n^4 \log{n}$ computation time. The time reported does not include the time for computing a tree decomposition, however there are many algorithms that can quickly compute a tree decomposition for many graphs of practical interest \citep[see e.g.,][]{tw-pid}.

We are able to compute the Shapley value for each of the covert networks in less than two minutes. For the 9/11 network, our computation took only 5.3 seconds. The Shapley value for this network has previously been approximated by \citet{approximation}, using a method based on a random sampling of 10.000 permutations of the players in the network. Lindelauf et al. report that this computation of approximate Shapley values took 5 minutes. Our method is not only exact, but also much faster.

Of course, the method of \citet{approximation} can be applied to any graph rather than just to graphs of small treewidth. However, it is not yet known how the performance of their approximation depends on the structure of the graph. Still, when the treewidth of the graph is small, our method provides an excellent way to compute exact Shapley values.

The IS in Europe network has treewidth 13. Despite this relatively high treewidth, our algorithm was still able to compute the Shapley value in 38.7 seconds. Our algorithm can thus handle graphs even with moderate treewidth quite quickly. It can also handle graphs with large numbers of vertices, although it appears from the results on the PACE networks that the polynomial factor in the running time ($n^4$) starts to dominate rather than the dependence on the treewidth.

\begin{table}[h]
  \begin{center}
  \footnotesize
  \begin{tabular}{| c | c | c | c | c | }
	\hline
	Graph & Treewidth & Vertices & Edges & Time (seconds) \\
	\hline
italian	&	3		&	65	&	113	&	0.6 \\ \hline
montreal	&	6		&	29	&	75	&	0.4 \\ \hline
9-11	&	8		&	69	&	163	&	5.3 \\ \hline
drugnet	&	8		&	193	&	273	&	119.4 \\ \hline
ise-extended	&	13		&	77	&	274	&	38.7 \\ \hline
pace\_005	&	5		&	201	&	253	&	31.1 \\ \hline
pace\_012	&	5		&	572	&	662	&	1746 \\ \hline
pace\_022	&	6		&	732	&	1084	&	22868 \\ \hline
pace\_023	&	6		&	990	&	1258	&	30255 \\ \hline
pace\_028	&	7		&	139	&	202	&	2262 \\ \hline
pace\_070	&	10		&	106	&	399	&	50.0 \\ \hline
  \end{tabular}
  \end{center}
  \caption{Performance of the algorithm on several real-world  networks and several benchmark graphs from the PACE 2018 challenge. For disconnected graphs, we considered only the largest connected component in the graph (for which the number of vertices and edges is given).}\label{tab:performance}
\end{table}

\FloatBarrier

\section{Conclusions}

Game-theoretic centrality measures are a powerful tool for identifying important vertices in networks. We have shown that, using treewidth, two game-theoretic centrality measures can be practically computed on graphs much larger than previously feasible, allowing us to analyze larger networks than before.

Our algorithm runs in time $2^{O(tw\log{tw})}n^{O(1)}$. The log-factor in the exponent is due to the need to keep track of a connectivity partition. A very interesting open question is whether the algorithm can be improved to have single-exponential running time, that is, is it possible to attain a $2^{O(tw)}n^{O(1)}$-time algorithm? For several (counting) problems involving connectivity, this is indeed possible: For instance, it is possible to count Hamiltonian Cycles or Steiner Trees in single-exponential time \citep{singleexp} by using approaches involving matrix determinants. Either a positive answer to this question or a conditional lower bound ruling out such an algorithm would be interesting.

We remark that the log-factor in the exponent represents only the worst case. However, since we are dealing with induced subgraphs, if two vertices share an edge, they can never be in two distinct connected components. Therefore, the actual number of connectivity partitions considered may be lower than suggested by the worst case bound. It would be interesting to see if it is possible to take this phenomenon into account when generating a tree decomposition: perhaps it would be possible to optimize a tree decomposition to limit the number of feasible partitions (for instance, by giving preference to bags that are cliques). Such an approach has previously been considered for Steiner Tree \citep{optimizetd}.

With a trivial adaptation, our algorithm can also be used to compute the Banzhaf value \citep{banzhaf} for $v^{conn}$ and $v^{wconn2}$; this requires merely a change in weighting values. 
Techniques similar to ours can also be used to evaluate other connectivity games, e.g., the Shapley value for $v^{wconn1}$ and $v^{wconn3}$ \citep{lindelauf13} can be computed by extending our notion of characteristic to also include the maximum weight of an edge in the subgraph induced by $S$.

Another interesting question is whether other connectivity measures can be computed using treewidth. For instance, $v^{conn}$ assigns a value of $0$ to any disconnected coalition, even if there exists a large connected component. It might be more reasonable to make the value of a coalition equal to the size of the largest connected component inside this coalition. It is easy to adapt our techniques to obtain an algorithm running in time $n^{O(tw)}$ for this case; it would be interesting to see if a fixed-parameter tractable algorithm exists.


\bibliographystyle{apalike}

\bibliography{references}


%
%


\end{document}